\documentclass[10pt,twocolumn,a4paper]{article}

\usepackage{graphicx,amsmath,amssymb,amsthm,mathabx,booktabs,txfonts}
\date{}

\author{Yong Tan\\~\\
\emph{yongtan\_navigation@outlook.com}}
\title{Compound Binary Search Tree and Algorithms}

\theoremstyle{plain}
\newtheorem{lemma}{Lemma}

\begin{document}
\maketitle 
\begin{abstract}
The Binary Search Tree (BST) is average in computer science which supports a compact data structure in memory and oneself even conducts a row of quick algorithms, by which people often apply it in dynamical circumstance. Besides these edges, it is also with weakness on its own structure specially with poor performance at worst case\cite{5}. In this paper, we will develop this data structure into a synthesis to show a series of novel features residing in. Of that, there are new methods invented for raising the performance and efficiency nevertheless some existing ones in logarithm or linear time.
\end{abstract}

\section{Introduction}
Binary Search Tree (BST) is a common data structure broad elaborated in many literatures and textbooks as that regular. At first, the construct on it can be referred to a binary tree in which besides each unit incident a \emph{key} (or value), each even carries three \emph{link}s mutually to comprise a compact structure, whose pointers respectively point to its own members in family the \emph{parent} and two \emph{children} that resides in the \emph{left} side and the \emph{right} side but maybe in\emph{null} for link in \emph{open}\cite{1,2,5}; especially the one without parent as \emph{root} or \emph{ancestor} to all others.

Inside a BST, analogous constitution may occur by generating roots and trees recursively\textemdash each can as minor \emph{root} on which new twigs can bloom from although there has existed a chief root to all items in tree. 

If refer to maintain or build a tree, which must comply a Protocol of constitution; say the least, all incident keys must obey the clause that each at left link or right link which in charge by its \emph{parent} should smaller than or larger than \emph{parent}'s. Consider within a more large rank, of two subtrees the left one or the right one and their common root, correspondently at the left or the right, each key of \emph{descendant} in tree is smaller or larger than ancestor’s.

Hence that law strongly conducts the operation of adding a fresh item into a BST, upon that, building a tree is actually accounted as a row of item insertions. Meanwhile, the single insertion can be outlined as a course of comparison as a path\cite{5}, called \emph{depth} by us. In theory, with the longest one among them, we can use to measure the shape of a BST. 

We can define a \emph{proper} tree with $\log{n}$ depth by such a \emph{bi}furcated structure above-mentioned on each item as a standard pattern, where variable $n$ is the number of items inside tree. For an accessing in a tree to achieve an operation, complexity can be estimated for lower bound in $\Theta(\log{n})$ or for upper bound in $O(n)$\cite{5}, clearly, both are decided by the shape of tree, frankly speaking, by a temporal series of insertions in building period.

Thus the flavor of cognition becomes interesting when we study the shape of BST: the future shape of BST has actually been destined by the permutation of that ready sequence in advance; in contrast, we either have not got any way to carry an arbitrary permutation suitable to guarantee the proper shape of building.

Worst still, that is a challenge to us so far; upon that, people turned to reduce the estimation of shape refer to a conception which surveys major likelihood of \emph{average} depth if the BST made up with a random sequence. In \cite{5} reported by the empirical of many retrials through, the average depth can into $2\ln n$ in most cases, which approximate to $1.39\log{n}$.

In substance, the conclusion Robert Sedgewick et al made is enough preferable to solve many estimations applied on algorithms executing on BST, at that in the book authors yet conceded those existing methods in a poor performance at worst case. In practice the tractability about manipulating a BST sometimes becomes vulnerable in some perhaps, at least to render the performance with instability.

\subsection{Results}(1) Develop the BST into a synthesis by integrating distinct structures and to survey those novel correspondences among them. (2) Discuss the operations involved so as to exert the advantages that have embedded in those new components and make them support mutually to raise whole performance. (3) Completely solve the issue of building a proper tree when with a stochastic sequence as input and guarantee the cost of building invested by a \emph{logarithm} time. (4) Estimate the batch works on BST than the traditional.

\subsection{Perliminated}
Of a tree, it must refer to a set $T=t_1,t_2,\ldots,t_n$ to denote all items in that tree where letter $n$ is the amount of items. We denote a depth by $\mathcal{D}.t_i$ for $t_i\in{T}$. In default, all keys in set $T$ are distinct one another unless other claim made, denote item and key with $\mathcal{K}.t_i $.

Excepted the \emph{terminal} in a tree without any child, the others called \emph{knot}, furthermore the \emph{Complete Knot} (CK) has two children in charge; another, the \emph{Partial Knot} (PK) with only one the left or the right.\\

\textbf{Contents Arrangement.} The 2nd section will prove some associated corollaries and introduce the new synthesis. The 3rd expands the essential methods on that structure. The 4th will get estimate the batch work on a representative pattern. The summary is arranged in last section.

\section{Morphology}
For item set $T$ on a BST, there may be a \emph{reference} (or \emph{ref}) set and further it conducts these two sets in a relationship of member mapping\textendash the pattern of \emph{bijection} so as to access that BST at each item can be done via the interface of \emph{ref} set. Furthermore this \emph{ref} set is permutable ordered by sorting their incident keys. Apparently of so doing, the correspondent permutation of \emph{ref}s should in a strictly ordered with all incident keys in ascent, for this, we can maintain a \emph{ref} set in a queue $(\rho_1,\rho_2,\ldots,\rho_n)$ denoted by $\mathcal{T}$; for two ordinals $0<i\leq j\leq n$ in queue, such that $\mathcal{K}.\rho_i\leq\mathcal{K}.\rho_j$, which is a prerequisite on this queue $\mathcal{T}$. 

We call the item in BST \emph{Preimage} relating to its ref. Moreover on a member $\rho_{s}\in\mathcal{T}$, the reach in $\mathcal{T}$ embraced by $\rho_{1}$ and $\rho_{s-1}$ we call \emph{Left Subqueue} or LS otherwise at right side the \emph{Right Subqueue} or RS by $\rho_{s+1}$ and $\rho_{n}$.

Without question, the queue $\mathcal{T}$, which is the axiomatic, as an invariant to BST\cite{5}; herein we call \emph{Adjoint Reference Axis} or ARA.

By the account of definition of an ARA, we can give a set of features on the relevance between two structures the ARA and the BST: For a pair of contiguous refs $\rho_{s}, \rho_{s+1}\in\mathcal{T}$ and with their preimages $t_{i}, t_{k}\in{T}$ respectively, some corollaries can be given as follow.

\begin{enumerate}
\item $\mathcal{D}.t_{i}(\rho_{s}) \neq\mathcal{D}.t_{k}(\rho_{s+1})$.
\item If $\mathcal{D}.\rho_{s}(t_{i})<\mathcal{D}.\rho_{s+1}(t_{k})$, then $t_{i}$ is an ancestor to $t_{k}$ otherwise inverse.
\item $t_{i}, t_{k}$ concurrently are neither as CKs nor as terminals.
\end{enumerate}
These proofs are as follows. 
\begin{proof}
Contrary to $\mathcal{D}.t_{i}(\rho_{s}) = \mathcal{D}.t_{k}(\rho_{s+1})$; refer to their depths which both from root to themselves, there must be a nonempty intersection on that both; moreover a member $t$ therein may as their common ancestor. Since at least there root of BST in that set, for oneself has been as the chief ancestor to all others, thus the form
\begin{equation}\label{ieq}\mathcal{K}.\rho_{s}(t_{i})<\mathcal{K}.\rho(t)<\mathcal{K}.\rho_{s + 1}(t_{k}).\end{equation}
should hold. That will be a contradiction to prerequisite of two refs $\rho_{s}$ and $\rho_{s+1}$ contiguous in ARA; or else case, items $t_i$ and $t_k$ both are same one. The \emph{first} holds and we call the inequality\eqref{ieq} The First Essential Correlation or TFEC. \\

\textbf{Second}, suppose $\mathcal{D}.t_{i}(\rho_{s}) > \mathcal{D}.t_{k}(\rho_{s+1})$. The \emph{second} could hold if $t_k$ as the root of BST. Assume $t_i$ not in the left subtree that subjects to root $t_k$, but they always have themselves to root a common ancestor $t$ according to the constitution of tree; further the TFEC holds for three $t_{i}, t_k$ and $t$ but it violates the prerequisite on ARA like the \emph{first}; another case that $t_k$ and $t$ be same one. Analogously for $\mathcal{D}.t_{i}(\rho_{s}) < \mathcal{D}.t_{k}(\rho_{s+1})$.\\

\textbf{Third}, (1) A terminal cannot be ancestor to another. (2) Contrary to $t_i$ and $t_k$ both as CKs, further suppose $t_k$ is an ancestor to $t_i$. 

Thus $t_i$ in the \emph{left} subtree that in charge by root $t_k$. Over there at least an item $t$ as child at $t_i$ \emph{right} link such that there is existence of $\mathcal{K}.t_{i}(\rho_{s})<\mathcal{K}.t$ and $\mathcal{K}.t_{i}(\rho_{s}), \mathcal{K}.t<\mathcal{K}.t_{k}(\rho_{s+1})$, then $t$'s ref will among $\rho_{s}$ and $\rho_{s+1}$ to lead to a contradiction. The \emph{third} holds.

\end{proof}

In 1979, J. H. Morris had invented a similar linear sequence made in $O(n\log{n})$ the lower boundary, he aimed to traverse a BST in a convenience\cite{4}. Robert Sedgewick et al even had projected the all items in a BST instance on a horizon in their book\cite{5} to render such a queue and sounded it is an invariant to a BST, moreover to detail something for that conception.

We here study the morphology of this synthesis that composed of two structures which thicker than theirs. The goal we longing is to attempt to find out some delicate features concealed inside unknown before, by which we can develop more quick algorithms than the existing and make two structures can support one another to reach a preferable execution upon manage data structure.

\section{Operations} 
\subsection{Deletion} 
The existing method of deleting an item off a BST is fairly in perplexity to people. The central issue is the deletion concerns the perhaps of damage to logical structure. In fact, the focus about this problem is on deleting a CK off, upon that it needs an alternate in the tree to charge the vacant position made by deleting; the selected condition at alternate clearly requires that one with the capable to hold on those trees again which ever rooted in that CK before. 

T. Hibbard in 1962 proposed the \emph{successor} as alternate\cite{5} which at the leftmost position in the \emph{right} subtree rooting CK which needs a progress of examination like $(>, <,\ldots,<)$ for seeking. Analogously to $(<, >,\ldots,>)$, Robert Sedgewick et al suggested the \emph{predecessor} in \emph{left} subtree the rightmost\cite{5}.

However for both routes, the cost at least is involved to $\log{n}$; In ARA model, we have an analysis about this selection as follows.

Deleting a \emph{terminal}, that may in $O(1)$ without overplus from other actions since no child in charge by terminal. For a PK with unique child, its child actually can as a root to hold on a subtree or null. However, the child can as the alternate for charge the vacant position that the deleted PK left down. At the worst case, the process will be \emph{deleting} one and \emph{moving} one, so the cost can in the $O(2)$ at the worst case. 

Therefore, of deleting a CK, if a terminal as the alternate, then equivalently to delete two and move one, the cost in $O(3)$ time; instead, for a PK as the alternate, the cost will become deleting two and move two. Eventually, the cost may in $O(4)$ at the worst case for deleting a CK.

Now in ARA model, we advise the two members in ARA with their preimages as alternates which by that \emph{CK}, either at the \emph{right} side or at the \emph{left} side, both contiguous to \emph{CK}; actually they are the \emph{successor} and \emph{predecessor} in tree to that CK which mentioned by Hibbard or Sedgewick; their opinions both are right.

Our advice may at least involves two grounds: (1) Either of two alternates in ARA, the key on it is the \emph{extremum} in LS (maximum) or RS (minimum) that means it has the qualification as new root to charge that tree which in the past to rooting CK, since the relationship for each key on descendant in that tree to new root likewise consists with that protocol of constitution of BST; the alternative whether as a parent or whether as an ancestor. (2) By our proof, there is not any likelihood for these alternates being CKs concurrently. 

Via the dimension of ARA we readily seek out the alternate so we gain a constant complexity that of $O(4)$ on deleting operation, which far less than logarithm.

\subsection{Insertion}
In contrast to deletion, insertion is more important that relating to building our synthesis that composed of a BST and an ARA; besides these, the insertion yet concerns to the function of offline manage a linear list. For example, the \emph{thread binary tree} invented by J. H. Morris in 1979\cite{4}, that can be referred to the result in executing an \emph{offline} method to obtain that list. When a set of \emph{online} accessing with frequent insertions upon that tree, the \emph{offline} will pay off a high price for a plenty of requirements of resorting.

What will changes happen in ARA when an insertion accomplishes in the BST? That we will survey is the key point that concerns if it in a proper tractability to us. The following proof will describe this evolution between the new item and its parent in BST, and their refs in ARA.

\begin{lemma}
As an item $t$ added into a BST as $t_{i}$'s child, consider their refs $\rho(t), \rho_{s}(t_i)\in\mathcal{T}$, the ref $\rho(t)$ will by $\rho_{s}(t_i)$ at the left side or the right side in ARA. 
\end{lemma}

\begin{proof}
If the fresh $t$ as \emph{left} child of its parent $t_{i}$, If $t_{i}$'s LS is empty in ARA, then the $\rho(t)$ will be interpolated by $\rho_{s}(t_i)$ at left side, the lemma holds. 

Instead, assume there is an item $\rho_{s-1}(t')$, the fresh $\rho(t)$ by it at the \emph{left} side in ARA. Since the fresh is a terminal in tree, then by those corollaries above-mentioned there should be an assertion come true: $t'$ is an ancestor to $t$. 

Also, the assertion will further lead to three cases about item $t_{i}$ and $t'$: (1) the $t'$ also is the ancestor to parent $t_i$; (2) or conversely; (3) $t_i$ and $t'$ are the same one. 

If the \emph{first} case holds, $t_{i}$ should have stayed in the \emph{right} subtree that in charge by root $t'$ for $\mathcal{T}.t_{i} >\mathcal{T}t'$, then it is clearly a contradiction the fresh impossible as the \emph{left} child of its parent since its parent in right subtree. Inversely for (2), then someone has been occupied the left link of $t_{i}$, maybe $t'$. 

Finally the \emph{third} case is truth\textendash $\rho(t)$ by $\rho_{s}(t_i)$ at left side. Analogously to fresh $t$ as $t_i$'s right child, the lemma holds.

\end{proof}

This lemma presents a clear correspondence that a sorted system with two dimensions the BST and the ARA. That not only makes our two building works concur on a routine to easy\textendash an adding operation in BST also being an insertion on ARA, both are fresh by its parent; on the other hand, this lemma has exposed the affinity of two structures: the BST can carry the information over ARA. In next subsection, we will exploit this feature. 

Herein, we lay the Doubly Linked List (DLL) on ARA as data structure to condition the dynamical insertion that may be caused from the frequently online. To the new data structure $T\oplus\mathcal{T}$, we call \emph{Compound Binary Search Tree} or CBST.

\subsection{A Simple Query}
Although ARA is a strictly sorted list by keys in ascent, yet there instantly appears to a challenge while an accessing on ARA merely with two ordinals, which try to obtain a piece-wise data like on a common sorted list whose items numbered by natural number; on account of such compact data structure and its kinds always with a fat chance in hashing or none the preferable to go.

At the aspect of maintaining a numbering system on an ARA, it never is none the easy: the incidence of point-wise renumbering in list which brought about by someone's change maybe reaches all corners through the whole; especially worse still for online algorithms than you imagine. For this query, we can convert the operation of location with ordinal to calculate the position in ARA. The new model will supply the maintainability in logarithm complexity.

Of an insertion in ARA meanwhile as being adding item into BST, we learn two means to extend an ARA along axis: one is adding member at the left or right end of a present ARA; another, the fresh one interpolated between two. So we define a structure in BST.

\textbf{Definition. }\emph{Given an item family composed of a \emph{grandpa}, a \emph{father} and a \emph{grandson}, we call the \emph{father} \emph{Flexed Node} (or FN), within their familial relationship, if and only if two links the \emph{grandpa}\textendash \emph{father} and \emph{father}\textendash \emph{ grandson}, both in distinct sides.}

For an instance: father at grandpa's \emph{left} link contrasted with grandson at father's \emph{right} link which to \textquotedblleft\emph{replicate}\textquotedblright~the path at father; the two links can comprise a \emph{Flexed Pipe} or FP; we call the left-right pattern of FP \emph{Clockwise} inversely $Anticlockwise$.

Thus, we can observe an interest process: \textquotedblleft Suppose an ordinal for ARA be known on \emph{grandpa}, there is a visit occurs repetitively along such bearing that at parent bound for child. If this progress reflected on ARA, it can render the ordinal in \emph{cumulation} or \emph{degression} upon that the visiting in forward or retreat among those refs which are two distinct bearings. For this \emph{one-step} the leaping does over those members in ARA, the ordinal’s change merely for correcting into one, we can reckon it as an \emph{invariance} addressing the fixed one on numbering ordinal. Instead, on a FN, the \emph{grandson}'s ref inserted among \emph{grandpa}'s and \emph{father}'s in ARA, the counting on ordinal should have to re-treat the one-point the \emph{grandson} into one. At this time, the reckoning on ordinal must be yielded to the changes on both the tree and the queue, because once a tree rooted in the \emph{grandson}, the treatment will be done with the \emph{variant} maybe many that relevant to the scalar of that tree other than singleton of invariance.\textquotedblright

Hence, we measure the case on that replication with variable \emph{Flexion} the number of items in that subtree; denoted by $\mathcal{F}$. Of that, a progressive visit along the clue of \emph{parent-child} in a tree actually evolves a leaping over an interval on that axis of ARA; the flexion can measure the thick of interval which embraced by \emph{grandpa} and \emph{father}.

We hence design the structure and functor in following, at first we suppose a visiting list $\phi=\varv_1, \varv_2,\ldots$ consisting of items that will be visited by functor $\varPhi$ and, the calculation is in the queue with ascending ordinal.\\

\textbf{The Structure}
\begin{enumerate}
\item \emph{For the \emph{root} of BST, let flexion $\mathcal{F}.\text{\emph{root}}$ equal of the number of items in its left subtree.}
\item \emph{The flexion on clockwise FN is of \emph{negative}; \emph{positive} for anticlockwise; none of the two into 0.}
\item \emph{If item at the left link as child, the fixed \emph{counting} for oneseft into $-1$, otherwise $1$; herein denoted by $\mathcal{S}$, especially $\mathcal{S}.\text{\emph{root}} = 1$.}
\end{enumerate}

The functor $\varPhi$ with visiting list $\phi$ will start at \emph{root} and assume the input ordinal is $N$:
\begin{align*}
\varPhi(\varv_{i+1}&) =\varPhi(\varv_{i}) + \mathcal{F}.\varv_{i+1}+\mathcal{S}.\varv_{i + 1};\\
\text{\textbf{s.t.} }&\left\{
\begin{array}{ll}
\varv_{1} = \text{root};~\varPhi(\varv_{0}) = 0 &\text{if}~i = 0;\\
\text{return}~\varv_{i}&\text{if}~\varPhi(\varv_{i}) = N;\\
\varv_{i+1} = \varv_{i}.\text{L\emph{child}}&\text{if}~\varPhi(\varv_{i}) > N;\\
\varv_{i+1} = \varv_{i}.\text{R\emph{child}}&\text{if}~\varPhi(\varv_{i}) < N;
\end{array}\right.
\end{align*}

The search of route on functor $\varPhi$ is likewise drawn from root to someone inside a tree step-by-step without distinct difference to an ordinary query. Therefore the lower bound of cost definitely involves to $\log{n}$. In addition to in ARA model, those intervals processed by functor $\varPhi$ in ARA incessantly shrinks over time, in this way the numeric can approach to the exact in the range as thin as possible till reach. 

The same course can be inversed to follow the clue of \emph{child-parent} applied on deletion certainly. 

Because in need of the function of ordinal query on CBST, these flexions incident FNs on that route that involving to the item has been deleted must be corrected with 1 or -1. Hereby the cost on deleting operation becomes involved to $\log{n}$ rather than ours above-mentioned.

We have introduced all principal operations in a CBST with FN model. Which these operations on each item, whether doing an \emph{insertion} or whether doing a \emph{query} or whether doing a \emph{deletion}, this model can always conduct them to obtain the ordinal involved in ARA simultaneously. 

On the other hand, it also maintains in a logarithm system a quick\textendash algorithm set. Meanwhile, we solve a challenge for a compact linear list with a well responsiveness in logarithm times, however on implementation or on maintaining. 

\subsection{Building}
We have got rid of the influence out from shape of BST in deletion, but the true of matter is not likelihood there for us to do some analogue things for other operations. Hence in this section we will discuss how to build a \emph{proper} BST from the dimension of ARA. That conducts building a CBST is not concurrently to construct two structures other than the insertion above-mentioned instead to fabricate it, of a rather manner of industrialization.

As building a Pyramid as following pseudo code show over there assume $n=2^{k}-1$, algorithm recursively extracts items from ARA as parents to connect with their children that have been reserved down. 

\begin{center}
\small{\textbf{Construct BST}}
\end{center}

\begin{flushleft}
/$\ast$~~\emph{Parameters}~~$\ast$/\\ $\eth = \xi= 2;$ ~\emph{//~the cursors, $\eth$ backup $\xi$}\\$\theta = 4;$~~\emph{//~offset for next extracted one} \\$\kappa = 1;$~~\emph{//~ the width between parent $\&$ children} \\
~\\
/$\ast$ $\ast$ $\ast$~~\emph{Building Module}~~$\ast$ $\ast$ $\ast$/\\
\small{
01. Loop($\xi<n$)\\
02. ~~$dl = \xi - \kappa; ~~ dr=\xi + \kappa;$~~\emph{// addresses for left $\&$ right}\\
03. ~~$\mathcal{T}(\xi).\text{L}link \coloneqq  \mathcal{T}(dl);~~\mathcal{T}(dl).Parent \coloneqq \mathcal{T}(\xi);$\\
04. ~~$\mathcal{T}(\xi).\text{R}link \coloneqq \mathcal{T}(dr);~~\mathcal{T}(dr).Parent \coloneqq \mathcal{T}(\xi);$ \\  
05. ~~$\xi\coloneqq \xi + \theta;$\\
06. ~~\textbf{if} $\xi> n$ and $\eta < n/2$ \\
07. ~~\textbf{then} $\kappa =\eth;~\eth=\xi= \theta; ~\theta\coloneqq 2\theta;$ \emph{// start next round}\\ 
08. \textbf{return} the tree $T$.
}
\end{flushleft}

In this algorithm, the process can swift convert the roles for items from parent to child. In every round, algorithm is equivalently to execute this converting module on an abstracted bed against the previous results. In this pattern, the primitive bed is ARA and, only the items on \emph{event}\textendash position in ARA can be permitted to participate.

For example, in the \emph{first} round, initially these items at 2\emph{nd}, 6\emph{th}, 10\emph{th}, $\ldots, (n - 1)th$ positions are involved in a new list as parents, where with \emph{offset} the argument equals of 4 to pick up \emph{parents}; on the other hand, with \emph{width} the argument equals of 1, the $odd$s are picked up as \emph{children}. 

Against the new sequence of parents, the \emph{second} round will do the similar performance which chooses the 4\emph{th}, 12\emph{th}, 20\emph{th}, $\ldots, (n - 3)th$ in ARA to as parents by \emph{offset} counted of 8, upon that those parents in \emph{first} round, 2\emph{nd}, 6\emph{th}, $\ldots$, now become children pointing to new parents, in which the \emph{width} argument equals of 2 always half of \emph{offset} argument.

In this way, the course will terminate when the \emph{middle} item at the $2^{k-1}$\emph{th} cups the rising \emph{pyramid} as chief root to whole, which inverses the course of existing building method that begins at root of BST.

As to ARA as a sorted linear list, the parameter $\kappa$ at 02\emph{th} step can utilize this feature to control the selection of children\textendash left element can but as left child, right child as well as left one; of that, this measure guarantees those keys comply the regulation about parent and children.

It wants us to answer a question about the detail of this kind of fabrication: There is any foul about each key in parent-child pattern, but not equivalent to in \emph{ancestor-descendent} one. For example, one key in \emph{right} subtree but smaller than the root's by which that subtree in charge. The cause is the foul member at that position in ARA ahead of the root but selected as child by a lower descendent in that subtree. We so far cannot confirm if our tactic of picking children is reliable to avoid this foul, specifically on using parameter $\kappa$.  In\cite{5}, Robert Sedgewick et al ever mentioned this kind of error.

The proof is simply to prove the extremum in a tree, its ref impossibly outs the cordon the position at which the ancestor stays. Hereby, we will only discuss the case\textendash ancestor versus its right subtree, thus it requires us prove the \emph{leftmost} item in subtree that ref’s position in ARA always at its ancestor’s right side.

Given a $j$th round (for $j>1$), and $t_{j}$ as parent in this round, so it with a right width $d_{R}(j) = 2^{j-1}$ for $t_{j}$ to pick up right child. Suppose a tree is rooted in $t_{j}$’s right child $t$ and, thus there is a \emph{path} containing $j-2$ links to catenate the left extremum and the root that right child $t$.

To measure the parameter $\kappa$'s change on every link in that path whose scalar can be quantified with the number of members that leaped over for picking child, which can be described by a series $\bigcup_{i=0}^{i\leq j-2}2^{i}$; among them each numeric has been applied as \emph{width} in every round to capture children before $j$th.

Well, a total number $d_{L}(j-1) = \sum_{i=0}^{i\leq j-2}2^{i}$ results in an interval the left extremum off the root the $t$. By binary addition, it is easy to compare $d_{L}(j-1)$ and $d_{R}(j)$ and results in $d_{R}(j) - d_{L}(j-1) =1$ deduced from the exponent on $d_{R}(j)$ more than $d_{L}(j-1)$ for one degree. That means the foul is impossible in PM. 

Meanwhile the proof does matter in another critical quality in PM that algorithm cannot pick up any item as parent or child repeatedly; because in very round, the candidates for in roots always outside those present trees that have been built up.

Of course, there is most likelihood to $n\neq 2^k - 1$ in practice; herein we can read it as $n=m+m'$ for $m=2^k-1$ and $n\slash{2}\leq {m}\leq{m}$; thus the \textquotedblleft\emph{overplus}\textquotedblright~$m'$ smaller than $m$. It is clear that these overplus ones can be settled at the most bottom as children before proceeding PM. The shape only just depends on user's tactic to arrange the positions for them maybe for balance. Anyway, the depth of CBST can be $\lfloor \log{n}\rfloor$ at all.

It is certainly that the PM works on a sorted sequence. Hereby we will introduce a sorting algorithm which developed from \emph{Tournament Method} that has been introduced in\cite{6} whose complexity has been known in $O(n\log n)$. We reformed it for condition the data structure the DLL, called \emph{Card Game Sorting Method} or CGSM\footnote{The source code and files involving to test this algorithm has been hosted in this website: \emph{https://github.com/snatchagiant/CGSM} which encoded by C++ and executed in console platform.}.

The pseudo code CGSM (1) about the \emph{engine} is in following.

\begin{center}
\small{\textbf{CGSM (1)}}
\end{center}

\begin{flushleft}
\small{Function: Insert$(s, ~t)$ \emph{// inserting in DLL, s precedes t}}\\
\small{Function: Follow$(s, ~t)$ \emph{// t follows s in DLL}}\\
~\\
/$\ast$ $\ast$ $\ast$ \small{Merger Method} $\ast$$\ast$ $\ast$ /\\
~\\
\small{Function: Merger($x,~ y$) \emph{// $x\in{X}; y\in{Y}$; the heads of queues} \\
01. $H = x;$ \textbf{if} $\mathcal{K}.x > \mathcal{K}.y$ \textbf{then} $H = y;$~~\emph{// elect the new head}\\
~\\
02. \textbf{Loop} ($y \neq\varnothing$)~~\emph{// not out the range of Y}\\
03. ~~\textbf{if} $\mathcal{K}.x < \mathcal{K}.y$ \\
04. ~~\quad\textbf{if} $x$ at the end of $X$\\
05. ~~\quad\textbf{then} Follow$(x, ~y)$; \textbf{break};~\emph{// follow x, the y $\&$ RS}\\
06. ~~\quad\textbf{else} $x\coloneqq x.\text{Later};$~~\emph{// continue on X.}\\
07. ~~\textbf{else} $\text{Insert}(y, ~x);$ $y\coloneqq y.\text{Later}$~~\emph{//insert $\&$ continue on Y}\\
08. \textbf{return} $H$;\\ 
}
\end{flushleft}

The sequence $Y$ can as well as sorted heap in ascent where the member that with the min key among all always springs out from the top of heap; for the outside member $y_{j}\in Y$, Merger Method attempts to seek out an appropriate interval in sequence $X$ for inserting it by moving the cursor in sequence $X$, whose process as sorting cards in card game. 

Which maintains the kernel logic, by a pair of two neighbors the $x_{i-1}, x_{i}\in{X}$ and member $y_{j}\in{Y}$ commonly carry the inequality $\mathcal{K}.x_{i-1}<\mathcal{K}.y_{j}<\mathcal{K}.x_{i}$  about three keys of theirs.

If functor to the end of $X$ and $Y$ no empty, then the rest subsequence in $Y$ would join to the right end of new $X$ together to compose the new sequence $X\biguplus{Y}$; than whole course end. Conversely none springing out $Y$ for search either leads to procedure terminate likewise.

\emph{Complexity. }Let $\vert X\vert = t$ and $\vert Y\vert =s$, we can describe the process of comparisons by a set $Y'$ as:
\[Y'=\bigcup^{s}_{i = 1}y_i\times{X_{i}}\colon{X_{i}}=x_k,\ldots, x_j\text{ for }{1\leq k, j\leq t}.\]
A Cartesian product in favor of present the detail on each member $y_i\in{Y}$. If a member $y_i$ inserted in sequence $X$ and ahead of member $x_{j}\in{X}$, than for shift to next member $y_{i+1}$, the proceeding will start at $x_{j}$. Thus there is a nonempty intersection $\vert{X_{i}}\cap{X_{i+1}}\vert=1$, of that at worst case it is equivalent to $(s+t)$ times of comparisons implemented on this algorithm which tallies the sum of scalars of two sequences the $X$ and the $Y$; therefore the complexity can in $O(n)$ where $n= \vert X\vert +\vert Y\vert$; likewise, each member in that two also may be consider as equivalently being invoked for precisely once.

To a random sequence with $n$th members, certainly, by a way to scan the sequence through, it is easy on the level of procedure to yield a group of components within the sequence where each with a sorted subqueue. The sorting job eventually becomes a multi sequences merging. 
The pseudo code CGSM (2) in following, we only relate another means that has the treatment on each member in sequence as a singleton set in order to exhibit the course of merging multitude of components based on the pairwise.

\begin{center}
\small{\textbf{CGSM (2)}}
\end{center}
\begin{flushleft}
/$\ast$~~\emph{Parameters}~~$\ast$/\\
\small{
$\Theta$; $\kappa= \vert \Theta\vert$;~~\emph{// store the heads of subsequences}.\\
$s = 1; ~ \pi = 0;$ \emph{// two cursors in $\Theta$}\\
~\\
/$\ast$~~\emph{Merger Rounds}~~$\ast$/\\
\textbf{Loop} ($\kappa > 1$)\\
01. $\pi++; \Theta[\pi] \coloneqq \text{Merger}(\Theta[s],~ \Theta[s + 1])$; $s\coloneqq s + 2$\\
02. \textbf{if} $s = \kappa$ \textbf{then} $\Theta[\pi] \coloneqq \Theta[s]$;~~\emph{// backup the last if $\kappa$ is odd}\\
03. \textbf{if} $s > \kappa$ \textbf{then} $\kappa = \pi; ~s=1; \pi = 0;$~~\emph{// start the next round}\\
}
\end{flushleft}

It is easy to count of $\lfloor\log{n}\rfloor$ rounds for whole merger and, moreover by the analysis above-mentioned, the cost can in $O(n\log{n})$.

Of course in this way,  many CBSTs merger can be looked like a process of CGSM, the cost for BSTs merger can be in $O(n\log\kappa+n)$ where variable $\kappa$ is the number of trees. This means let we gets rid of much more annoying troubles that results from a mass of relations intertwined by tree's shape.

In addition to about concurrently building a FN system in PM, we give a solution for a module in algorithm which in following.

\begin{flushleft}
/$\ast$~~\emph{Structure}~~$\ast$/\\
\small{$t\in{T};~t.\ell = 0;~t.r = 0;$~~\emph{// two counters for left $\&$ right.}}\\
~\\
 /$\ast$~~\emph{Module}~~$\ast$/\\
 \small{
 01. \textbf{As} \emph{Parent} \textbf{then} \emph{record the number of items in subtree.}\\
 02. $t.\ell = \lambda.\ell +\lambda.r + 1$;\quad\emph{// $\lambda$ is the left child.}\\
 03. $t.r = \rho.\ell + \rho.r + 1$;\quad\emph{// $\rho$ is the right child.}\\
 04. \textbf{As} \emph{Child} ~~\textbf{if} $\mathcal{K}.t< \mathcal{K}.\pi$\quad\emph{// $t$ at the parent $\pi$ left link}\\
 05. \textbf{then} $\mathcal{F}.t = (-1)\ast{\varv.r}; ~~\mathcal{S}.t  = -1;$\quad\emph{// clockwise}\\
 06. \textbf{otherwise} $\mathcal{F}.t = \varv.\ell; ~~\mathcal{S}.t  = 1;$~~\quad\emph{// anticlockwise.}\\
 }
\end{flushleft}

There is another alternative, the building course may begin at the \emph{median} position in ARA in manner of \emph{top to bottom} completely, by recursively bisection sequence to work out the BST. Here we don't intend to have a length to introduce, over there it may occasion the tree depth in $\lceil\log{n}\rceil + 1$.

\section{Bacth of Works}
It is apparently that we have reduced more trivial and unnecessary steps in our Merger method contrasted with \emph{merge sort} introduced in \cite{5}. Hence we naturally propose a theme the batch of works on CBST since the ARA also as another dimension supporting BST where they might in equivalence. Thus we can look the algorithm on two sequences in the manner of \emph{many-many} other than the traditional method which everyone in a fixed sequence in turn invoked for accessing on the whole BST, the character of \emph{one-many}. 

Of course, in some occasion, such as one-time locking database for a bunch of jobs may spare much more resources than many-time ones. As a theoretical discussion, we merely reduce the issue simply to the preferable or not by the way of complexity analysis which is applicability. So we aim that: (1) Mark off the boundary for two methods, the \emph{batch} and the \emph{traditional} if instance in proper shape. (2) What is the index in a BST? By this guide we learn which alternative is more preferable with tree being inharmonic. (3) Analyze the instance of locality of accessing.

Here we only discuss the case of \emph{query} inasmuch with others they are similarly one another. We firstly let sequence $X=x_{1},\ldots,x_{n}$ as an ARA; otherwise, refer sequence $Y=y_{1},\ldots,y_{\kappa}$ to as the query sequence. 

Than we add a module in Merger method and have a bit of reforming. The routine will process a \emph{success hit}\cite{5} that referred to a query in BST with a ref, the ref succeeds to match the item inside that tree; now here a member in sequence $Y$ instead of that ref, it matches a member in set $X$ as well as $\mathcal{K}.x_{i}=\mathcal{K}.y_{j}$. Functor will return \emph{yes} to sound the success hit.

Conversely, with no match and concurrently $\mathcal{K}.x_{i}>\mathcal{K}.y_{j}$ come true, it states the member $y_{j}$ \emph{failure} for match; no hit happen. These additions have barely to increase the overall complexity nevertheless have added a conditional for execution.

Let $\kappa = \lambda{n}$ (for $0<\lambda\leq{1}$). In the case of sequence $Y$ been sorted in ascent, the cost of the \emph{batch} could be in $(1+\lambda)n$ approach to $O(2n)$. 

Contrast to the \emph{traditional} in the lower bound $\lambda{n}\log{n}$ with members in set $Y$ in turn for query on BST; if $\lambda{n}\log{n}\geq{(1+\lambda)n}$, we have \emph{boundary} $\lambda\geq\log^{-1}{n\slash{2}}$ such that \[\lambda=\rho^{-1} \text{ for } \rho=\lceil\log{n}\rceil.\] If the scalar of set $Y$ beyond, the batch is worthwhile.

Consider plus a sorting on sequence $Y$ in $\kappa\log\kappa$, as the measure with a delicate difference to $\kappa\log{n}$ the traditional, our task hence has to be altered to estimate the shape of BST. 

Assume the cost is $\lambda{n}\hbar$ of using traditional query in CBST where $\hbar$ is the depth of tree. Thus an inequality
\[\lambda{n}\hbar\geq(1+\lambda){n} + \lambda{n}\log \lambda{n}= n+\lambda{n}\log{2}\lambda{n}.\] 
Eventually $\hbar= \lambda^{-1}+\log2\lambda{n}$. Consider $\varg=\hbar - \log{n}$ then $\varg = \lambda^{-1}+\log{2}\lambda{n} -\log{n}$.

Such that $\varg=1 + \lambda^{-1} - \log\lambda^{-1}$. For $0<\lambda\leq 1$, having $\log\lambda^{-1}\lll\lambda^{-1}$ so $\varg > 0$, means the existence of the depth $\hbar$ in a CBST can as an index to measure that tree. For example, $\lambda = 1/4$ then $\hbar\geq{4+\log{(2n\slash{4})}}$, further if $\hbar\geq\log{8n}$, we can execute the batch. 

It is interesting that if two extremums with extreme keys in sequence $Y$  in use to reduce a subqueue in ARA for query, where is boundary? 

At first, the variable $\theta{n}$ ($0\leq\theta\leq{1}$) represents the scalar of that queue locked up in ARA, then the following equation marks out the boundary which has involved the cost of sorting on sequence $Y$
\[\lambda{n}\log{n} - \lambda{n}\log\lambda{n} = (\theta + \lambda)n.\] 
Where the cost for the queries that with two extreme keys for lock a field inside a proper CBST, it can be considered to be ignored as negligible quantity. 

Thus we have the boundary $\theta = \lambda\log(2\lambda)^{-1}$. When $\lambda = 1$ the $\theta$ into \emph{null}, the batch none the worthy; when $\lambda = 1\slash{2}$ the $\theta$ into 0 that means the traditional still worthwhile; if $\lambda\leq{1\slash{4}}$, then $\theta$ always larger than or equal of $\lambda$ by a constant-fold. The boundary indeed is rather volatile.

\section{Summary}
It is interest to develop a common data structure BST into a synthesis. In fact, more and more features covert in the CBST model needs us to reveal, such as maintain a proper BST at a lower cost or further reform the structure for special purpose for clients. On the other hand, we have improved the situation of poor performance at worst case on BST better than before to serve database management.

\end{document}